\documentclass[conference]{IEEEtran}
\usepackage[bottom]{footmisc}
\usepackage[para]{threeparttable}
\usepackage{graphicx}
\usepackage{amsmath}
\usepackage{varwidth}

\usepackage{subcaption}
\usepackage{graphicx}
\usepackage{amsthm}

\usepackage[noend]{algpseudocode}
\usepackage{algorithm}

\newtheorem{definition}{Definition}

\newtheorem{theorem}{Theorem}

\usepackage{tikz}
\usetikzlibrary{plotmarks}
\usepackage{pgfplots}
\usetikzlibrary{backgrounds}
\usetikzlibrary{patterns}
\usepgfplotslibrary{units}

\usepackage{epsfig}
\usepackage{url}
\usepackage{stfloats}
\usepackage{enumerate}
\usepackage{longtable}

\usepackage{multirow}
\usepackage{latexsym}
\usepackage{verbatim}
\usepackage{amssymb}

\usepackage{xspace}

\usepackage{xcolor}      

\newcommand{\eat}[1]{}                

\newcommand{\zo}{{\{0,1\}}}

\renewcommand{\sf}{\ensuremath{s_{id}}}

\newcommand{\A}{$\mathcal{A}$~}

\newcounter{linecounter}

\newcommand{\Ra}{\ensuremath {\stackrel{\$}{\leftarrow}{\xspace}}}
\newcommand{\as}{\ensuremath {\leftarrow}{\xspace}}

\usepackage{enumitem}
\usepackage{pifont}

	\newcommand{\sk}{\ensuremath {\mathit{sk}}{\xspace}}
\newcommand{\pk}{\ensuremath {\mathit{pk}}{\xspace}}

\newcommand{\eda}{\ensuremath {\texttt{\textbf{ARIS}}}{\xspace}}

\newcommand{\edakg}{\ensuremath {\texttt{\textbf{ARIS}.Kg}}{\xspace}}
\newcommand{\edasig}{\ensuremath {\texttt{\textbf{ARIS}.Sig}}{\xspace}}
\newcommand{\edaver}{\ensuremath {\texttt{\textbf{ARIS}.Ver}}{\xspace}}

\newcommand{\Fq}{\ensuremath {\mathbb{F}_q}{\xspace}}
\newcommand{\EC}{\ensuremath {E(\Fq) }{\xspace}}
\newcommand{\Zp}{\ensuremath {\mathbb{Z}_p}{\xspace}}
\newcommand{\prf}{\ensuremath {\texttt{PRF}}{\xspace}}
\newcommand{\hash}{\ensuremath {\texttt{H}}{\xspace}}
 \newcommand{\B}{\ensuremath{\mathcal{B}}{\xspace}}

  \newcommand{\LL}{\ensuremath{\mathcal{L}}{\xspace}}
\usepackage{cleveref}

\usepackage{wrapfig}
\usepackage{capt-of}
\usepackage{bold-extra}
\usepackage{multirow}
\newcommand{\specialcell}[2][c]{%
	\begin{tabular}[#1]{@{}c@{}}#2\end{tabular}}

\renewcommand\IEEEkeywordsname{Keywords}
\newcommand\blfootnote[1]{%
	\begingroup
	\renewcommand\thefootnote{}\footnote{#1}%
	\addtocounter{footnote}{-1}%
	\endgroup
}

\def\endthebibliography{%
	\def\@noitemerr{\@latex@warning{Empty `thebibliography' environment}}%
	\endlist
}

\begin{document}

\title{\eda: Authentication for Real-Time IoT Systems}

\IEEEoverridecommandlockouts


\author{\IEEEauthorblockN{Rouzbeh Behnia$^\star $} \thanks{$^\star$Work done in part when Rouzbeh Behnia and Muslum Ozgur Ozmen were at Oregon State University.}
\IEEEauthorblockA{University of South Florida\\
	Tampa, Florida \\
	behnia@mail.usf.edu}
\and
\IEEEauthorblockN{Muslum Ozgur Ozmen$^\star $}
\IEEEauthorblockA{University of South Florida\\
	Tampa, Florida \\
	ozmen@mail.usf.edu}
\and

\IEEEauthorblockN{Attila A. Yavuz}
\IEEEauthorblockA{University of South Florida\\
Tampa, Florida \\
attilaayavuz@usf.edu}}

\maketitle

\begin{abstract}
Efficient authentication is  vital for IoT applications with stringent  minimum-delay requirements (e.g., energy delivery systems). This requirement becomes even more crucial when the  IoT devices are battery-powered, like small aerial drones, and the efficiency of authentication directly translates to more operation time.  Although some fast authentication techniques have been proposed, some of them might not fully meet the needs of the emerging delay-aware IoT. 


In this paper, we propose a new  signature scheme called \eda~that pushes the  limits of the existing  digital signatures, wherein a commodity hardware can verify 83,333 signatures per second. \eda~also enables the fastest signature generation along with the lowest energy consumption and end-to-end delay among its counterparts. These significant computational advantages come with a larger storage requirement, which is a favorable trade-off for some critical delay-aware applications. These desirable features are achieved by harnessing message encoding with cover-free families and a special elliptic curve based one-way function.  We prove the security of \eda~under the hardness of the elliptic curve discrete logarithm problem in the random oracle model. We provide an open-sourced implementation of \eda~on commodity hardware and 8-bit AVR microcontroller for public testing and verification.
\blfootnote{© 2019 IEEE. Personal use of this material is permitted. Permission from IEEE must be obtained for all other uses, in any current or future media, including reprinting/republishing this material for advertising or promotional purposes, creating new collective works, for resale or redistribution to servers or lists, or reuse of any copyrighted component of this work in other works.}

\end{abstract}
\vspace{1mm}
\begin{IEEEkeywords} Authentication; Internet of Things; digital signatures;  delay-aware systems; applied cryptography.  \end{IEEEkeywords}

 \section{Introduction}\label{sec:Introduction}

 IoT systems often need authentication  for applications that need to verify  a large volume of incoming transactions or commands. While   symmetric key primitives (e.g., HMAC) can provide very fast authentication, they fail to offer non-repudiation which is often vital for these applications. For instance, Visa handles   millions of transactions every day \cite{VisaCard}.  Each transaction corresponds to multiple authentications of the user's request and card information on merchant's side, payment gateway and credit card issuer \cite{CreditCardinfra}.  Therefore,  creating more efficient solutions can significantly reduce the overall authentication overhead of such systems that results in substantial financial gains. 
 
The need for efficient authentication becomes even more imperative for applications in which IoT devices must operate in safety-critical settings and/or with battery limitations. For instance, battery-powered aerial drones \cite{Drone:Elisa:Won:2015} might communicate and authenticate streams of commands and measurements with an operation center in  a short period of time. A fast and energy-efficient authentication can improve the flight and response time of such aerial drones \cite{Dronecrypt}. Other IoT applications such as smart grid  systems, which involve battery-powered sensors, will also  benefit from fast and energy-efficient digital signatures  which minimize the   authentication delay/overhead and improve the operation time of  the  sensors~\cite{SmartGrid:PMU:Experiment:Journal:2017}.
Additionally,  in vehicular networks, safety significantly hinges  on the  end-to-end delay~\cite{IEEE1609.2_SecServices}, and therefore attaining  a signature scheme with the lowest end-to-end delay is always desired.

\subsection{Our Contributions}

In this paper, we propose a new efficient signature scheme called \eda.  
\eda~makes use of an Elliptic Curve Discrete Logarithm Problem (ECDLP) based  one-way  function   and exploits the homomorphic properties of such functions to (i) linearly add the private key elements to attain a shorter signature and (ii) mask this addition with a one-time randomness $ r $ to achieve a (polynomially-bounded) multiple-time signature scheme.  
We outline the main properties of \eda~as below. 
\begin{itemize}

\item \underline{\em Fast Verification:} 
\eda~provides the fastest signature verification among its counterparts. More specifically, \eda~pushes the limits of   elliptic curve (EC) based signature schemes by providing  nearly $ 2\times  $ faster verification as compared to its fastest counterpart \cite{SchnorrQ}.  

\item \underline{\em Fast Signing:}  The signature generation of \eda~avoids  expensive computations such as  fixed-base scalar multiplication. Therefore, \eda~achieves 33\% faster signing as compared to its fastest counterpart \cite{SchnorrQ}. 

\item \underline{\em Low End-to-End Delay:} 
Due to having the fastest signature generation and  verification algorithms, \eda~achieves nearly $ 40\% $ lower end-to-end delay, as compared to its fastest counterpart \cite{SchnorrQ}.  This might encourage the potential  adoption of  \eda~for applications that require delay-aware authentication.

\item \underline{\em Energy Efficiency:} 
 By avoiding any computationally expensive operation in the signing and verification algorithms, \eda~achieves the lowest energy consumption as compared to its state-of-the-art efficient counterparts.  Specifically,  as shown in Figure \ref{fig:sign}, the verification algorithm in \eda~attains $ 40\% $ lower energy consumption as compared to its most energy efficient counterpart. This makes \eda~potentially suitable for IoT applications wherein the battery-powered devices authenticate telemetry and commands (e.g., aerial drones).

\item \underline{\em Tunable Parameters:}  \eda~enjoys from a highly tunable set of parameters. This allows \eda~to be instantiated with different properties for different applications. For instance, the parameters set that we considered for our implementation on AVR microcontroller enjoys from a smaller public key and private key pair, and if the same scheme is implemented on commodity hardware, it can enjoy from a faster signature generation ($ 2\times $ faster than the scheme in \cite{SchnorrQ})  by incurring a few microseconds on the verification algorithm. 
\end{itemize}
\noindent\textbf{Limitations: } All of the desired properties and efficiency gains in \eda~come with the cost of  larger key sizes. For instance, in the verification efficient instantiation of \eda~(as in Table \ref{tab:Laptop}), which has the largest key sizes, the size of the public key and private key could be as large as $ 32$KB. However, this can be   decreased to  $  16$KB and $ 8 $KB for the private key and public key sizes (respectively) while still maintaining the fastest signature generation and verification algorithms among its counterparts.  We have shown that even with these parameters sizes, \eda~can be implemented on 8-bit AVR  while enjoying from the most computation and energy efficient algorithms as shown in  Figure \ref{fig:sign}, Figure \ref{fig:verify} and Table \ref{tab:AVR}.

\begin{figure*} 
	\centering
	\begin{minipage}[b]{.45\textwidth}
		\includegraphics[width=1\linewidth]{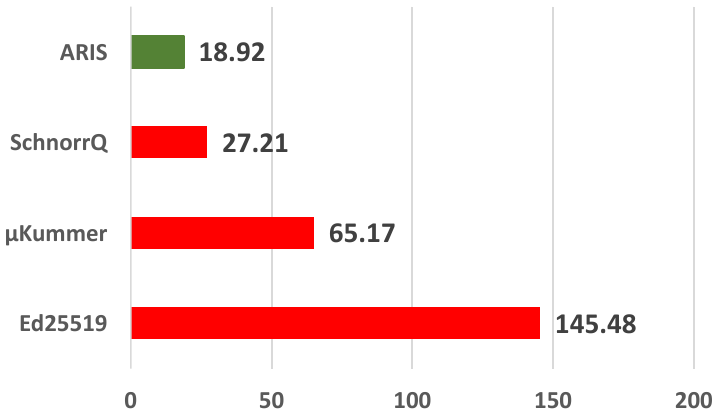}
		\caption{Energy consumption (mJ) for signature generation of \eda~and its counterparts on AVR microcontroller}	\label{fig:sign}
	\end{minipage}\qquad
	\begin{minipage}[b]{.45\textwidth}
		\includegraphics[width=1\linewidth]{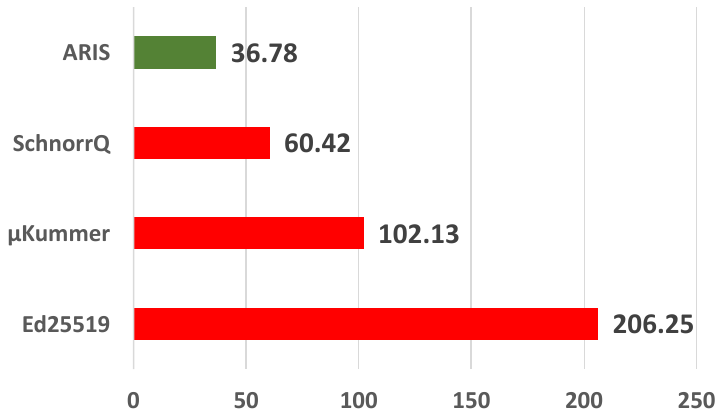}
		\caption{Energy consumption (mJ) for signature verification of  \eda~and its counterparts on AVR  microcontroller}	\label{fig:verify}
	\end{minipage}
	
\end{figure*}
\section{Related Work}\label{sec:related}
 
 One-time signatures (e.g. HORS \cite{HORS_BetterthanBiBa02})  have been proposed to offer fast signing and verification.  Following HORS,  many schemes with different performance and security trade-offs such as time valid one-time signatures (i.e., TV-HORS  \cite{TVHORSInfocom09}) have been proposed. However, these schemes suffer from security and performance penalties incurred due to the need for time-synchronization and their low tolerance for packet loss.  Multiple-time hash-based signatures (e.g., XMSS \cite{XMSS:Buchmann:2011}) utilize Merkle-Tree and can sign multiple messages by keeping the signer's state. Recently, stateless variations (e.g., SPHINCS  \cite{SPHINCS}) have been proposed, however such schemes suffer from large signatures ($ \approx $ 41 KB) and slow signing algorithms.  
%

Recently, a polynomially-bounded multiple-time signature scheme based on HORS design  is proposed \cite{Tachyon}.  The scheme utilizes the additive homomorphic property of the underlying one-way function to obtain fast signatures  where the signer only aggregates private key components  during the online phase.  However, despite its   efficiency,  it cannot meet the stringent delay requirement of some IoT applications. 
Another  proposed scheme called CEDA \cite{CEDA}  exploits the aggregatable property of RSA-based one-way permutation functions and message encoding (as proposed in \cite{HORS_BetterthanBiBa02}) to  attain efficient signing. However, the large parameter sizes not only incur very large public keys but also make the  exponentiations that takes place during signature generation and verification quite costly. Therefore CEDA, while being among the most efficient schemes, does not surpass the latest implementations of signatures on fast elliptic curves. 

In the line of  proposing fast elliptic curves, Renes et al. \cite{Kummer} presented an  efficient instantiations of  the scheme in \cite{Ed25519} based on Kummer surface that shows significant performance gains as compared to its base scheme \cite{Ed25519}.   In 2016, Costello et al.  \cite{SchnorrQ} proposed a new implementation of  \cite{Ed25519} based on another  elliptic curve called  FourQ  which shows to even outperform the implementation in  \cite{Kummer}.   

\section{Preliminaries}\label{sec:Prelim}

\noindent \textbf{Notation.} Given two primes $ p $ and $ q $ we define a finite field \Fq~and a group \Zp.  We also work on \EC~as an elliptic curve over   \Fq. We commonly denote $ P\in \EC $ as a generator of the points on the curve.  $ x\Ra S $ denotes randomly selecting      $ x  $ from a set $ S $. We denote scalars as small letters (e.g., $ x $) and points on curve as capital letters (e.g., $ P $).   We denote tables/matrices as bold capital letters (e.g., $ \mathbf{P} $). We define the bit-length of a variable  as $ |x| $, i.e., $ |x| = \log_2 x $. Scalar and point multiplication is denoted as $ xP $.  We define two Pseudo Random Functions 			$\texttt{PRF}_1: \{0,1\}^*\rightarrow \Zp$  and 	$\texttt{PRF}_2: \{0,1\}^*\rightarrow \{0,1\}^\kappa$  and three hash function $\hash_1:\{0,1\}^* \times \Zp \rightarrow \{0,1\}^{l_1}$, $\hash_2:\EC \rightarrow \{0,1\}^{l_2}$, and $\hash_3:\{0,1\}^* \times \{0,1\}^{l_2} \rightarrow \{0,1\}^{l_1}$  for some integers $ l_1 $ and $ l_2 $, to be defined in Section \ref{sec:Performance}.

\begin{definition} \label{def:ECDLP}\normalfont ({\em Elliptic Curve Discrete Logarithm Problem}) For \EC~as an elliptic curve over a finite field \Fq,  given  $ P,Q\in \EC  $,  the Elliptic Curve Discrete Log Problem (ECDLP) asks  to find $ k  \in \Zp$,  if  it exists, such that  $ Q=kP $.
 
\end{definition}

\begin{definition}\normalfont
	A  signature scheme consists of three algorithms $  \texttt{SGN} = (\texttt{Kg,Sig,Ver}) $ defined as follows.
	
	\begin{itemize}\normalfont
		\item[--] $ (\sk, \pk) \leftarrow \texttt{SGN.Kg}(1^\kappa) $: Given the security parameter $ \kappa $, it outputs the private and public key pair $ (\sk,\pk) $.
		
		\item[--] $ \sigma \leftarrow \texttt{SGN.Sig}(m,\sk) $:
		Given the message   $ m $ and the signer's private key  $\sk$, it outputs the signature $ \sigma $.
		\item[--] $ \{0,1\} \leftarrow \texttt{SGN.Ver}(m, \sigma, \pk) $: Given a message-signature pair   ($m,\sigma $), and the claimed signer's public key  $\pk$, it outputs a decision bit  $ d \gets \{0,1\} $.
	\end{itemize}
\end{definition}

In the following definition, we define the security of signature schemes based on the methodology proposed in \cite{Bellare-2006}.  After the initialization phase i.e., $ \texttt{SGN.Kg}(\cdot) $,  The adversary \A  is given access to the signature generation oracle. \A  wins, if it outputs a \emph{valid} message-signature pair (that was not previously outputted from the sign oracle)  after making polynomially-bounded number of queries. 

\begin{definition}\label{def-EUCMA}\normalfont
Existential Unforgeability under Chosen Message Attack (EU-CMA)  experiment $ \mathit{Expt}^{\mathit{EU-CMA}}_{\mathtt{SGN}} $   is defined as follows.
\begin{itemize}\normalfont
\item [--] $ (\sk,\pk)\leftarrow \texttt{SGN.Kg}(1^\kappa) $
\item[--] $ (m^*, \sigma^*)\leftarrow  \mathcal{A}^{\texttt{SGN.Sig}(\cdot)} (\pk) $
\item[--]  If $  1 \leftarrow {\texttt{SGN.Ver}(m^*, \sigma^*, \pk)} $ and $ m^* $ was not queried to $ \texttt{SGN.Sig}(\cdot) $, return 1, else, return 0.

\end{itemize}

The EMU-CMA advantage of    \A     is defined as $ \mathit{Adv}^{\textit{EU-CMA}}_{\texttt{SGN}} =  \Pr[ \mathit{Expt}^{\mathit{EU-CMA}}_{\texttt{SGN}}= 1]$.
\end{definition}

\section{Proposed Scheme}\label{sec:proposed}

\eda~leverages the  homomorphic property of its underlying ECDLP-based one-way function, which is due to the exponent product of powers property, to achieve (polynomially-bounded)  multiple-time signatures from the one-time signature scheme proposed in \cite{Reyzin2002}, with   more compact signatures.  More specifically, in \eda, the private key consists of $ t $ randomly generated values $ x_i $ (generated using a $ \kappa $ bit seed $ z $) and the corresponding public key consists of all $ Y_i \gets  x_i P $ for  $ i \in \{1,\dots,t\} $.   

 To sign a message,  the signer obtains $ k $ indexes $ (i_1, \dots, i_k) $ by hashing the message (and a random input), uses the indexes  $ (i_1, \dots, i_k) $ to retrieve the corresponding private key elements (i.e., $  x_{i_j}$ where $ j \in \{1,\dots,k\} $) and sums them along with a one-time randomness $ r $.  The signature consists  of $ s $ and $ h $, which is obtained by applying the hash function $ \hash_2(\cdot) $  on $ R $,  that is computed as the output of applying the  one-way  function 
 on the one-time randomness $ r $.  
 
 Verification takes place by computing the summation of  the corresponding public key elements  (i.e., $  Y_{i_j}$)  and their subtraction from the output of the ECDLP-based one-way function applied on  $ s $.  The verifier outputs \emph{valid} if the subtraction  yields the same value of $ R $ as computed in the signature generation. Additionally, \eda~uses the BPV method in~\cite{BPV:basepaper:1998} to convert  an EC scalar multiplication to only $ k $ (where  $ k =18 $ or $k=28$ for our proposed parameter sets) EC point additions with the cost of storing a small, constant-size table.

Our scheme consists of the following algorithms.  
\noindent  $\underline{(\sk,\pk)\as \edakg(1^{\kappa})}$: Given the security parameter $ \kappa  $, this algorithm selects parameters $ (t,k) $  such that $     \binom{t}{k} \geq 2^\kappa$ and $z  \Ra \Zp $ and works as follows.
 \begin{enumerate} 
   
 \item Compute $ x_i \gets \prf_1(z,i)$  and  $ Y_i  \gets x_iP \mod  q $ for $ i\in \{1,\dots,t\} $ and set $ \mathbf{Y} \gets \{Y_i\}_{i=1 } ^{ t} $.
 \item Compute $ r_i \gets \prf_2(z,i)$  and  $ R_i  \gets r_iP \mod  q  $ for $ i\in \{1,\dots,t\} $ and set $ \mathbf{R} \gets \{R_i\}_{i=1 } ^{ t} $. 
 	\item Output    $   pk \gets   \mathbf{Y} $ and  $\sk \gets (z,\mathbf{R})   $ as the public key and private key, respectively.
 \end{enumerate}

\noindent $\underline{\sigma\as \edasig(m,\sk)}$: Given a message $ m\in \{0,1\}^* $ to be signed, this algorithm works as follows.
\begin{enumerate}
	\item Compute $ (i'_1,\dots,i'_k) \gets \hash_1(m,z)$ where $ |i'_j| \leq |t| $ for $ j\in \{1,\dots,k\} $.
	\item Compute $ r_{i'} \gets \prf_2(z,i'_j)$  for $ j\in \{1,\dots,k\} $, set $ r \gets \sum_{i'=1}^{k} r_{i' } \mod p $.
	\item Retrieve $ R_{i'} \gets \mathbf{R}[i'_j] $  for $ j\in \{1,\dots,k\} $,  compute $ R \gets  \sum_{i'=1}^{k} R_{i'}\mod  q $ and $ h\gets \hash_2(R) $.
	\item Compute $ (i_1,\dots,i_k) \gets \hash_3(m,h)$ (where $ |i_j| \leq |t| $)  and $ x_{i}\gets \prf_1(z,i_j)$  for $ j\in \{1,\dots,k\} $. 
	\item  Compute  $ s \gets r  - \sum_{i=1}^{k}x_{i} \mod p $ and output $ \sigma \gets (s,h) $.
 	
\end{enumerate}

 \noindent $\underline{\zo\as \edaver(m,\sigma,\pk)}$: Given a message-signature pair $  (m,\sigma) $ and $ pk  $, this algorithm works as follows.
 \begin{enumerate}
 	\item Parse $ (s,h) \gets   \sigma $ and compute $(i_1,\ldots,i_k) \as \hash_3(m,h)$,   where $ |i_j| \leq |t| $ for $ j\in \{1,\dots,k\} $.
 	\item Retrieve   $ Y_i\gets \mathbf{Y} [i_j] $ for $ j\in \{1,\dots,k\} $) and set $ Y \gets \sum_{i=1}^{k} Y_i \mod q  $.
 	\item Compute $ R' \gets sP + Y  \mod q  $ and check if $ \hash_2(R')  = h$ holds output \emph{valid}, and \emph{invalid } otherwise. 

 \end{enumerate}

\section{Security Analysis}\label{sec:security}
 We prove that \eda~is EU-CMA secure,  as defined in Definition \ref{def-EUCMA}, in the Random Oracle Model (ROM) \cite{Bellare:1993}. The proof   uses the Forking Lemma \cite{Bellare-Neven:2006}. 
 
 \begin{theorem}
In the ROM, if  adversary \A can  $(q_S,q_H) $-break the EU-CMA  security of  \eda~after making $ q_H $ and $ q_S $ random oracles and signature queries, respectively; then we can build another algorithm \B~that runs \A as a subroutine and can solve an instance of the ECDLP (as in Definition \ref{def:ECDLP}).
 \end{theorem}

\begin{proof}
We let $ Y^*\Ra \EC $  be an instance the  ECDLP  for algorithm \B~to solve. On the input of $ Y^* $ and $ z\Ra \Zp $, \B~works as follows. 

\noindent\underline{\emph{Setup:}} \B~keeps three lists $ \mathcal{L}_i $ for $ i\in\{1,2,3\} $ to keep track of the outputs of the random oracles and a list  $ \mathcal{L}_m $ to store the messages submitted to the sign oracle. \B~sets up the random oracle $RO$-$Sim(\cdot) $ to handle the hash functions and generates the users' public keys as follows.
\begin{itemize}
 \item \emph{Setup} $RO$-$Sim(\cdot)$: \B~implements $RO$-$Sim(\cdot) $ to handle queries  to hash functions $ \hash_1,\hash_2 $ and $ \hash_3 $, which are modeled as random oracles, as follows. 
 
 \begin{enumerate}
 	\item  $ \alpha_1 \gets  RO$-$Sim(m,z, \mathcal{L}_1)$: If $(m,z) \in   \LL_1$, it returns the corresponding value $ \alpha_1 $. Else, it returns $ \alpha_1 \Ra \{0,1\}^{l_1} $ as the answer and adds $  (m,z,\alpha_1) $ to $ \LL_1 $.
 	 	\item  $ \alpha_2 \gets  RO$-$Sim(R, \mathcal{L}_2)$: If $ R \in   \LL_2$, it returns the corresponding value $ \alpha_2 $. Else, it returns $ \alpha_2 \Ra \{0,1\}^{l_2} $ as the answer and adds $  (R,\alpha_2) $ to $ \LL_2 $.
 	  	 	\item  $ \alpha_3\gets  RO$-$Sim(m,h, \mathcal{L}_3)$: If $ (m,h) \in   \LL_3$, it returns the corresponding value $ \alpha_3 $. Else, it returns $ \alpha_3\Ra \{0,1\}^{l_1} $ as the answer and adds $  (m,h,\alpha_3) $ to $ \LL_3 $.	
 	 
 \end{enumerate}
 \item \emph{Setup Public Key}: Given the parameters $ (p,q,P,t,k) $, \B~works as follows to generate the user public key. 
  \begin{enumerate}
 	\item  Select $ j\Ra [1,t]$ and sets the challenge public key element $ Y_j \gets Y^* $.
 	\item Generate $ x_i \Ra  \Zp $ for $ i\in \{1,\dots,t\} $ and $ i\neq j $. 
 	\item Compute $ Y_i \gets x_i P $ for $ i\in \{1,\dots,t\} $ and $ i\neq j $.
 	\item Set $ \sk \gets \{x_i\}_{i=1,i\neq j} ^t $ and $ \pk \gets \{Y_1,\dots,Y_t\}$.
 	
 \end{enumerate}
\end{itemize}

\noindent\underline{$ \mathcal{A} $\emph{'s Queries:}} \A queries the hash functions $ \hash_i $ for $ i\in\{1,2,3\} $ and the sign oracle   for up to $ q_H $ and $ q_S $ times, respectively. \B~works as follows to handle these queries. 
\begin{itemize}
  \item \emph{Hash Queries:} $ \mathcal{A}$'s queries to hash functions $ \hash_1,\hash_2 $ and $ \hash_3 $ are handled by the $RO$-$Sim(\cdot) $ function described above.
  \item \emph{Signature Queries:}  \B~works as follows to answer $ \mathcal{A}$'s  signature query on message $ m $. If $ m\in \LL_m $, \B~retrieves the corresponding signature from $ \LL_m  $ and returns to $ \mathcal{A}$. Else, if $ m\notin \LL_m $, it works as follows.  

 \begin{enumerate}
	\item   Select $ s\Ra \Zp$ and compute $ S\gets sP $.
	\item Select $ k $   indexes $ (i_1,\dots i_k) \Ra [1,\dots,t]$.
 	\item  Set $ R  \gets  S - \sum_{i=1}^{k}Y_i$ and $ \alpha_2 \gets \{0,1\}^{l_2} $ and add $ (R,\alpha_2)  $ to $ \LL_2 $.
 	\item  If $  (\langle i_1,\dots i_k\rangle,h) \in \LL_3 $ abort. Else, add $  (m,h,\langle i_1,\dots i_k\rangle) $ to $ \LL_3 $. 
 	\item  Output $ \sigma =(s,h) $ to $ \mathcal{A}$ and add  $ (m,\sigma)\in\LL_m $.
\end{enumerate}

\end{itemize}
    \noindent\underline{$ \mathcal{A} $\emph{'s Forgery:}} Eventually, \A outputs a forgery  $ \sigma^* =(s^*,h^*) $ on message $ m^* $ and public key $ pk  $. Following the EU-CMA definition (as in Definition \ref{def-EUCMA}), \A only wins the game if $ \edaver(m^*,\sigma^*,pk) $ returns $ valid $ and $ m^* $ was never submitted to signature queries in the previous stage (i.e., $ m^*\notin \LL_m $).
    
\noindent\underline{\emph{Solving the ECDLP:}}  If \A does not output a valid forgery before making $ q_H $ hash queries and $ q_S $ signature queries, \B~also fails to solve the instance of ECDLP. Otherwise, if \A outputs a valid forgery    $ ( m^*,\sigma^* =\langle s^*,h^*\rangle ) $, using the forking lemma, \B~rewinds \A with the same random tape as in \cite{Bellare-Neven:2006}, to get a second forgery $  ( m',\sigma' =\langle s',h'\rangle ) $ where, with an overwhelming probability  $ s^* \neq s'  $ and $ h^*  = h' $. Based on  \cite[Lemma 1]{Bellare-Neven:2006}, $ \hash_3(m^*,h^*) \neq \hash_3(m^*,h') $, therefore, given $ (m^*,h^*) \in \LL_3$ and $ (m^*,h') \in \LL_3$,  \B~can solve a random instance of the ECDLP problem (i.e., $ Y^* $) if one of the following conditions hold. 
 \begin{itemize}
\item \textbf{Case 1:} For   $(i^*_1,\dots,i^*_k)  \gets \hash_3(m^*,h^*) $  and  $(i'_1,\dots,i'_k) \gets  \hash_3(m^*,h')    $   we have $ j\in (i^*_1,\dots,i^*_k) $ and $ j\notin (i'_1,\dots,i'_k) $.
\item  \textbf{Case 2:} For   $(i^*_1,\dots,i^*_k)  \gets \hash_3(m^*,h^*) $  and  $(i'_1,\dots,i'_k) \gets  \hash_3(m^*,h')    $   we have $ j\notin (i^*_1,\dots,i^*_k) $ and $ j\in (i'_1,\dots,i'_k) $.
\end{itemize}
If any of the above cases holds, \B~works as follows. If Case 1 holds, $ x_j \gets s^* -\sum_{\eta=1,\eta\neq j}^{k} x_{i^*_\eta}-s'-\sum_{\eta=1}^{k} x_{i'_\eta}\mod p$. Else, if Case 2 hold, $ x_j \gets s' -\sum_{\eta=1,\eta\neq j}^{k} x_{i'_\eta}-s^*-\sum_{\eta=1}^{k} x_{i^*_\eta}\mod p$. \end{proof}
\section{Performance Evaluation}\label{sec:Performance}
 
We have fully implemented \eda~on Four$ Q $ curve \cite{FourQ} which is known to be the fastest EC that provides 128-bit of security. We provide   implementations of \eda~ on both  commodity hardware and 8-bit microcontroller to evaluate its performance since most IoT applications are comprised  of them both (e.g., commodity hardware as servers or control centers and microcontrollers as IoT devices connected to sensors).  We compare the performance of \eda~with state-of-the-art digital signature schemes on both of these platforms, in terms of computation, storage and communication. Our implementation is open-sourced at the following link.

\begin{center}
	\url{https://github.com/rbehnia/ARIS} 
\end{center}




\subsection{Performance on Commodity Hardware}

\subsubsection{Hardware Configurations} We used a laptop equipped with Intel i7 Skylake processor @ $2.60$ GHz and $12$ GB RAM. 

\subsubsection{Software Libraries} We implemented \eda~using the open-sourced Four$ Q $ implementation~\cite{FourQ}, that offers the fastest EC operations, specifically EC additions that is critical for the performance of \eda. We used an Intel processor as our commodity hardware and leveraged Intel intrinsics to optimize our implementation. Specifically, we implemented our $\prf$ functions with Intel intrinsics (AES in counter mode). We used   blake2 as our hash function~\cite{blakeHash} due to its efficiency. 

We ran the open-source implementations of our counterparts on our hardware to compare their performance with \eda.

\subsubsection{Parameter Choice} Since we implement \eda~on Four$ Q $ curve, we use its parameters given in~\cite{FourQ}, which provide $128$-bit security. Other than the curve parameters, the choice of $t,k$ also plays a crucial role for the security of \eda. Specifically, $k$-out-of-$t$ combinations should also provide $128$-bit security to offer this level of security overall. On the other hand, we  can tune these parameters to achieve  our desired security level with different performance trade-offs. If we increase $t$ and decrease $k$, this results in a larger storage with faster computations,  and vice versa. For our commodity hardware implementation, we choose $t=1024$ and $k=18$, that we believe offers a reasonable trade-off between storage and computation as well as offering the desired $128$-bit security level. We set $ l_1 = 180 $ and $ l_2 =256 $.

\subsubsection{Experimental Results} We present the results of our experiments in Table \ref{tab:Laptop}. We observe that \eda~offers very fast signature generation and verification. It only takes 9 microseconds to generate a signature and 12 microseconds to verify it. This is the fastest among our counterparts, where the closest is SchnorrQ. Furthermore, if we use the same parameters set as for the AVR microcontroller, we can further speed up the signature generation to $ 6.5 $ microseconds, with the cost of a few microseconds on the verification speed. In SchnorrQ, a scalar multiplication is required in signature generation and a double scalar multiplication in verification. In \eda, EC additions are required for signature generation and verification is done with a scalar multiplication and EC additions. This corresponds to a $33\%$ faster signature generation and $83\%$ faster verification for \eda, compared to SchnorrQ. Therefore, we believe \eda~can be an ideal alternative for real-time applications.

\eda~signature size is the same with its EC-based counterparts \cite{ECDSA,Ed25519,Kummer,SchnorrQ} , that is significantly lower than its RSA-based and hash-based counterparts \cite{RSA,CEDA,SPHINCS}. On the other hand, \eda~comes with a larger private and public key, that is $32$ KB. 

\begin{table*}[t!]
	\centering
 	\vspace{ 2mm}

	\caption{Experimental performance comparison of \eda~and its counterparts on a commodity hardware} \label{tab:Laptop}
	\begin{threeparttable}
		\begin{tabular}{| c || c | c | c | c |  c | c | }
			\hline
			\textbf{Scheme} & \specialcell[]{\textbf{Sign Generation}\\  \textbf{ Time (}$\mu$s\textbf{)}} & \specialcell[]{\textbf{Private Key}\textsuperscript{$ \dagger$} \\ \textbf{(KB)}} & \specialcell[]{\textbf{Signature }\\  \textbf{Size (KB)}} & \specialcell[]{\textbf{Signature Verification}\\  \textbf{ Time (}$\mu$s\textbf{)}} & \specialcell[]{\textbf{Public Key} \\ \textbf{(KB)}} & \specialcell[]{\textbf{End-to-End} \\ \textbf{Delay (}$\mu$s\textbf{)}} \\ \hline \hline 
			
			
			SPHINCS~\cite{SPHINCS} & $13458$ & $1.06$ & $41000$ & $370$ & $1.03$ & $13828$\\ \hline
				TACHYON~\cite{Tachyon} & $138$ & $0.016$ & $4.4$ & $18$ & $864$ & $156$\\ \hline
			RSA~\cite{RSA}  & $8083$ & $0.75$ & $0.41$ & $48 $& $0.38$ & $8131$ \\ \hline 
			
			CEDA~\cite{CEDA}  & $55$ & $0.41$ & $0.41$ & $115$ & $384.38$ & $170$ \\ \hline 
			
			
			ECDSA~\cite{ECDSA}  & $725 $& $0.03$ & $0.06$ & $927$ & $0.03$ & $1652$\\ \hline
			
			Ed25519~\cite{Ed25519} &  $132$ &$ 0.03$ & $0.06$ & $335$ & $0.03$ & $467$ \\ \hline
			
			Kummer~\cite{Kummer} &  $23$ & $0.03$ & $0.06$ & $38$ & $0.03$ & $61$ \\ \hline 
			
			SchnorrQ~\cite{SchnorrQ} &  $12$ & $0.03$ & $0.06$ & $22$ & $0.03$ & $34$ \\ \hline \hline
			
			\eda  & $\mathbf{9}$ & $32.03$ & $0.06$ & $\mathbf{12}$ & $32$ & $\mathbf{21}$ \\ \hline 			
			
%
%
			
		\end{tabular}
		\begin{tablenotes}[flushleft]\scriptsize{  
				
				\item $ \dagger $ System wide parameters (e.g., p,q,$\alpha$) for each scheme are included in their corresponding codes, and private key size denote to specific private key size.
				
				
			}
		\end{tablenotes}
	\end{threeparttable}
\end{table*}
\begin{table*} [t]
	\centering
	\caption{Experimental performance comparison of \eda~and its counterparts on 8-bit AVR microcontroller} \label{tab:AVR}
	\begin{threeparttable}
		\begin{tabular}{| c || c | c | c | c |  c | c | }
			\hline
			\textbf{Scheme} & \specialcell[]{\textbf{Signature Generation}\\  \textbf{ Time (}s\textbf{)}} & \specialcell[]{\textbf{Private Key} \\ \textbf{(KB)}} & \specialcell[]{\textbf{Signature }\\  \textbf{Size (KB)}} & \specialcell[]{\textbf{Signature Verification}\\  \textbf{ Time (}s\textbf{)}} & \specialcell[]{\textbf{Public Key} \\ \textbf{(KB)}} & \specialcell[]{\textbf{End-to-End} \\ \textbf{Delay (}s\textbf{)}} \\ \hline \hline

			
			
			ECDSA~\cite{ECDSA} & $1.77$ & $0.03$ & $0.06$  & $1.80$ & $0.03$ & $3.57$ \\ \hline
			
			Ed25519~\cite{Ed25519, Ed255198bit} & $1.45$ & $0.03$ & $0.06$  & $2.06$ & $0.03$ & $3.51$  \\ \hline 
			
			$\mu$Kummer~\cite{Kummer, Kummer8bit} &  $0.65$& $0.03$ & $0.06$  & $1.02$& $0.03$  & $1.67$  \\ \hline 
			
			SchnorrQ~\cite{SchnorrQ, FourQ8bit} & $0.27$& $0.03$ & $0.06$  & $0.60$& $0.03$  & $0.87$\\ \hline \hline
			
			\eda  & $\mathbf{0.19} $ & $16$ & $0.06$ & $\mathbf{0.37}$ & $8$ & $\mathbf{0.56}$ \\ \hline
			
		\end{tabular}
		\begin{tablenotes}[flushleft] \scriptsize{
				
			} 
			
		\end{tablenotes}
	\end{threeparttable}
	 \vspace{-4mm}
\end{table*}

\subsection{Performance on 8-bit AVR}

\subsubsection{Hardware Configurations} We used an 8-bit AVR ATmega 2560 microcontroller as our IoT device to implement \eda. ATmega 2560 is equipped with $256$ KB flash memory, $8$ KB SRAM and $4$ KB EEPROM, with a maximum clock frequency of $16$ MHz. ATmega 2560 is extensively used in practice for IoT applications (especially in medical implantables) due to its energy efficiency~\cite{ATmega2560Medical}.

\subsubsection{Software Libraries} We implemented \eda~on ATmega 2560 using the 8-bit AVR implementation of Four$ Q  $ curve~\cite{FourQ8bit}, that provides the basic EC operations and a blake2 hash function. We implemented our scheme with IAR embedded workbench and used its cycle-accurate simulator for our benchmarks. 

As for our counterparts, we used their open-sourced implementations~\cite{Ed255198bit,Kummer8bit,FourQ8bit,microECC}. Note that we only compare \eda~with its EC-based counterparts, due to their communication and storage efficiency. Moreover, resource-contrained processors such as ATmega 2560 may not be suitable for heavy computations (e.g., exponentiation with 3072-bit numbers in RSA \cite{RSA}  and CEDA \cite{CEDA}). 

\subsubsection{Parameter Choice} As mentioned, \eda~can be instantiated with different $t,k$ values that offers a trade-off between storage and computation. Since ATmega 2560 is a storage-limited device, we select our parameters as $ t = 256$ and $k = 28$ to offer storage efficiency. Moreover, this allows us to store the private components ($x_i$ and $r_i$), instead of deterministically generating them at signature generation, and still have a tolerable storage even for an 8-bit microcontroller.  We also set $ l_1 \gets 224 $ and $ l_2 \gets 256 $.


\subsubsection{Experimental Results} Table \ref{tab:AVR}~shows the performance of \eda~compared with its  counterparts. The speed improvements of \eda~can also be observed for ATmega 2560. \eda~is $42\%$ faster in signature generation and $76\%$ faster in signature verification compared to its closest counterpart \cite{SchnorrQ}. This can translate into a significant practical difference when considered real-time applications that require fast authentication. Note that these benchmarks are obtained with a more ``storage friendly'' parameter choice, and can be further accelerated with different parameter choices where the microcontroller  is not memory-constraint.

One may notice that due to our parameter choice, the key sizes in our 8-bit microcontroller implementation are smaller. As aforementioned, this is because we select a different parameter set for $t,k$. Moreover, we store the private components as well, that correspond the 8 KB of the signer storage. Since we store these keys on the flash memory of ATmega 2560, they only correspond to $6\%$ and $3\%$ of the total memory, for private key and public key, respectively. Therefore, although we have significantly larger keys than our EC-based counterparts, it is still feasible to store them even on highly resource-constrained 8-bit microcontrollers.

\subsubsection{Energy Efficiency} It is highly desirable to minimize the energy consumption of cryptographic primitives in IoT applications to offer a longer battery life. For microcontrollers, energy consumption of the device can be measured with the formula $E = V * I * t$, where $V$ is voltage, $I$ is current and $t$ is the computation time~\cite{Precomputation:LowCostSig:Ateniese:NDSS2013}. Considering that the voltage and the current of a microcontroller are constant when the device is active, the energy consumption linearly increases with the computation time. Since \eda~offers the fastest signature generation and verification, energy consumption of \eda~is the lowest among its counterparts, and therefore would be preferred in applications that require longer battery life.


 \section{Conclusion}\label{sec:Conclusion}
In this paper, we presented a new efficient signature scheme to meet the strict minimum  delay requirements of  some real-time IoT systems.  This is achieved by harnessing the homomorphic property of the underlying ECDLP-based one-way function  and the precomputation technique proposed in \cite{BPV:basepaper:1998}. Our experimental results showed that the proposed scheme outperforms its state-of-the-art counterparts in signing and verification speed  as well as in energy efficiency.  The proposed scheme is shown to be secure, in the Random Oracle Model, under the hardness of the ECDLP. We   open-sourced our implementation  to enable public testing and verification. 
\vspace{1mm}

\noindent \textbf{Acknowledgment. }  This work is supported by  the Department of Energy award DE-OE0000780 and NSF award  \#1652389.

\bibliographystyle{IEEEtran}
 \bibliography{crypto-etc}

\begin{thebibliography}{10}
\providecommand{\url}[1]{#1}
\csname url@samestyle\endcsname
\providecommand{\newblock}{\relax}
\providecommand{\bibinfo}[2]{#2}
\providecommand{\BIBentrySTDinterwordspacing}{\spaceskip=0pt\relax}
\providecommand{\BIBentryALTinterwordstretchfactor}{4}
\providecommand{\BIBentryALTinterwordspacing}{\spaceskip=\fontdimen2\font plus
\BIBentryALTinterwordstretchfactor\fontdimen3\font minus
  \fontdimen4\font\relax}
\providecommand{\BIBforeignlanguage}[2]{{%
\expandafter\ifx\csname l@#1\endcsname\relax
\typeout{** WARNING: IEEEtran.bst: No hyphenation pattern has been}%
\typeout{** loaded for the language `#1'. Using the pattern for}%
\typeout{** the default language instead.}%
\else
\language=\csname l@#1\endcsname
\fi
#2}}
\providecommand{\BIBdecl}{\relax}
\BIBdecl

\bibitem{VisaCard}
\BIBentryALTinterwordspacing
J.~Steele. (2018) Debit card statistics. [Online]. Available:
  \url{https://www.creditcards.com/credit-card-news/debit-card-statistics-1276.php}
\BIBentrySTDinterwordspacing

\bibitem{CreditCardinfra}
\BIBentryALTinterwordspacing
O.~Papadimitriou. (2009) How credit card transaction processing works: Steps,
  fees \& participants. [Online]. Available:
  \url{https://wallethub.com/edu/credit-card-transaction/25511/}
\BIBentrySTDinterwordspacing

\bibitem{Drone:Elisa:Won:2015}
J.~Won, S.-H. Seo, and E.~Bertino, ``A secure communication protocol for drones
  and smart objects,'' in \emph{Proceedings of the 10th ACM Symposium on
  Information, Computer and Communications Security}, ser. ASIA CCS '15.\hskip
  1em plus 0.5em minus 0.4em\relax ACM, 2015, pp. 249--260.

\bibitem{Dronecrypt}
M.~O. Ozmen and A.~A. Yavuz, ``Dronecrypt - an efficient cryptographic
  framework for small aerial drones,'' in \emph{Milcom 2018 Track 3 - Cyber
  Security and Trusted Computing (Milcom 2018 Track 3)}, Los Angeles, USA,
  2018.

\bibitem{SmartGrid:PMU:Experiment:Journal:2017}
T.~Tesfay and J.~Y.~L. Boudec, ``Experimental comparison of multicast
  authentication for wide area monitoring systems,'' \emph{IEEE Transactions on
  Smart Grid}, vol.~PP, no.~99, 2017.

\bibitem{IEEE1609.2_SecServices}
``{IEEE} standard for wireless access in vehicular environments security
  services for applications and management messages,'' \emph{IEEE Std
  1609.2-2013 (Revision of IEEE Std 1609.2-2006)}, pp. 1--289, April 2013.

\bibitem{SchnorrQ}
C.~Costello and P.~Longa, ``Schnorrq: Schnorr signatures on fourq,'' MSR Tech
  Report, 2016. Available at: https://www. microsoft.
  com/en-us/research/wp-content/uploads/2016/07/SchnorrQ. pdf, Tech. Rep.,
  2016.

\bibitem{HORS_BetterthanBiBa02}
L.~Reyzin and N.~Reyzin, ``Better than {BiBa}: Short one-time signatures with
  fast signing and verifying,'' in \emph{Proceedings of the 7th Australian
  Conference on Information Security and Privacy ({ACIPS '02})}.\hskip 1em plus
  0.5em minus 0.4em\relax Springer-Verlag, 2002, pp. 144--153.

\bibitem{TVHORSInfocom09}
Q.~Wang, H.~Khurana, Y.~Huang, and K.~Nahrstedt, ``Time valid one-time
  signature for time-critical multicast data authentication,'' in
  \emph{{INFOCOM} 2009, {IEEE}}, April 2009.

\bibitem{XMSS:Buchmann:2011}
J.~Buchmann, E.~Dahmen, and A.~H\"{u}lsing, ``{XMSS} - a practical forward
  secure signature scheme based on minimal security assumptions,'' in
  \emph{Proceedings of the 4th International Conference on Post-Quantum
  Cryptography}, Berlin, Heidelberg, 2011, pp. 117--129.

\bibitem{SPHINCS}
D.~J. Bernstein, D.~Hopwood, A.~H{\"u}lsing, T.~Lange, R.~Niederhagen,
  L.~Papachristodoulou, M.~Schneider, P.~Schwabe, and Z.~Wilcox-O'Hearn,
  ``Sphincs: Practical stateless hash-based signatures,'' in \emph{Advances in
  Cryptology -- EUROCRYPT 2015}, E.~Oswald and M.~Fischlin, Eds.\hskip 1em plus
  0.5em minus 0.4em\relax Springer Berlin Heidelberg, 2015, pp. 368--397.

\bibitem{Tachyon}
R.~Behnia, M.~O. Ozmen, A.~A. Yavuz, and M.~Rosulek, ``Tachyon: Fast signatures
  from compact knapsack,'' in \emph{Proceedings of the 2018 ACM SIGSAC
  Conference on Computer and Communications Security}, ser. CCS '18.\hskip 1em
  plus 0.5em minus 0.4em\relax New York, NY, USA: ACM, 2018, pp. 1855--1867.

\bibitem{CEDA}
M.~O. Ozmen, R.~Behnia, and A.~A. Yavuz, ``Compact energy and delay-aware
  authentication,'' in \emph{2018 IEEE Conference on Communications and Network
  Security (CNS)}, 2018, pp. 1--9.

\bibitem{Kummer}
D.~J. Bernstein, C.~Chuengsatiansup, T.~Lange, and P.~Schwabe, ``Kummer strikes
  back: New dh speed records,'' in \emph{Advances in Cryptology -- ASIACRYPT
  2014}, P.~Sarkar and T.~Iwata, Eds.\hskip 1em plus 0.5em minus 0.4em\relax
  Springer Berlin Heidelberg, 2014, pp. 317--337.

\bibitem{Ed25519}
\BIBentryALTinterwordspacing
D.~J. Bernstein, N.~Duif, T.~Lange, P.~Schwabe, and B.-Y. Yang, ``High-speed
  high-security signatures,'' \emph{Journal of Cryptographic Engineering},
  vol.~2, no.~2, pp. 77--89, Sep 2012. [Online]. Available:
  \url{https://doi.org/10.1007/s13389-012-0027-1}
\BIBentrySTDinterwordspacing

\bibitem{Bellare-2006}
M.~Bellare and P.~Rogaway, ``The security of triple encryption and a framework
  for code-based game-playing proofs,'' in \emph{Advances in Cryptology -
  EUROCRYPT 2006}, S.~Vaudenay, Ed.\hskip 1em plus 0.5em minus 0.4em\relax
  Springer Berlin Heidelberg, 2006, pp. 409--426.

\bibitem{Reyzin2002}
L.~Reyzin and N.~Reyzin, ``Better than biba: Short one-time signatures with
  fast signing and verifying,'' in \emph{Information Security and Privacy: 7th
  Australasian Conference, ACISP Proceedings}, L.~Batten and J.~Seberry,
  Eds.\hskip 1em plus 0.5em minus 0.4em\relax Springer Berlin Heidelberg, 2002,
  pp. 144--153.

\bibitem{BPV:basepaper:1998}
V.~Boyko, M.~Peinado, and R.~Venkatesan, ``Speeding up discrete log and
  factoring based schemes via precomputations,'' in \emph{Advances in
  Cryptology --- EUROCRYPT'98: International Conference on the Theory and
  Application of Cryptographic Techniques Proceedings}.\hskip 1em plus 0.5em
  minus 0.4em\relax Springer Berlin Heidelberg, 1998, pp. 221--235.

\bibitem{Bellare:1993}
M.~Bellare and P.~Rogaway, ``Random oracles are practical: A paradigm for
  designing efficient protocols,'' in \emph{Proceedings of the 1st ACM
  Conference on Computer and Communications Security}, ser. CCS '93.\hskip 1em
  plus 0.5em minus 0.4em\relax New York, NY, USA: ACM, 1993, pp. 62--73.

\bibitem{Bellare-Neven:2006}
M.~Bellare and G.~Neven, ``Multi-signatures in the plain public-key model and a
  general forking lemma,'' in \emph{Proceedings of the 13th ACM Conference on
  Computer and Communications Security}, ser. CCS '06.\hskip 1em plus 0.5em
  minus 0.4em\relax NY, USA: ACM, 2006, pp. 390--399.

\bibitem{FourQ}
C.~Costello and P.~Longa, ``Four${Q}$: Four-dimensional decompositions on a
  ${Q}$-curve over the mersenne prime,'' in \emph{Advances in Cryptology --
  ASIACRYPT 2015}, T.~Iwata and J.~H. Cheon, Eds.\hskip 1em plus 0.5em minus
  0.4em\relax Springer Berlin Heidelberg, 2015, pp. 214--235.

\bibitem{blakeHash}
\BIBentryALTinterwordspacing
J.-P. Aumasson, L.~Henzen, W.~Meier, and R.~C.-W. Phan, ``Sha-3 proposal
  blake,'' Submission to NIST (Round 3), 2010. [Online]. Available:
  \url{http://131002.net/blake/blake.pdf}
\BIBentrySTDinterwordspacing

\bibitem{ECDSA}
\emph{{ANSI X9.62-1998:} Public Key Cryptography for the Financial Services
  Industry: The Elliptic Curve Digital Signature Algorithm ({ECDSA})}, American
  Bankers Association, 1999.

\bibitem{RSA}
R.~Rivest, A.~Shamir, and L.~Adleman, ``A method for obtaining digital
  signatures and public-key cryptosystems,'' \emph{Communications of the
  {ACM}}, vol.~21, no.~2, pp. 120--126, 1978.

\bibitem{Ed255198bit}
M.~Hutter and P.~Schwabe, ``Nacl on 8-bit avr microcontrollers,'' in
  \emph{Progress in Cryptology -- AFRICACRYPT 2013}, A.~Youssef, A.~Nitaj, and
  A.~E. Hassanien, Eds.\hskip 1em plus 0.5em minus 0.4em\relax Springer Berlin
  Heidelberg, 2013, pp. 156--172.

\bibitem{Kummer8bit}
J.~Renes, P.~Schwabe, B.~Smith, and L.~Batina, ``$\mu$kummer: Efficient
  hyperelliptic signatures and key exchange on microcontrollers,'' in
  \emph{Cryptographic Hardware and Embedded Systems -- CHES 2016}, B.~Gierlichs
  and A.~Y. Poschmann, Eds.\hskip 1em plus 0.5em minus 0.4em\relax Springer
  Berlin Heidelberg, 2016, pp. 301--320.

\bibitem{FourQ8bit}
Z.~Liu, P.~Longa, G.~C. C.~F. Pereira, O.~Reparaz, and H.~Seo,
  ``Four$\mathbb{Q}$ on embedded devices with strong countermeasures against
  side-channel attacks,'' in \emph{Cryptographic Hardware and Embedded Systems
  -- CHES 2017}, W.~Fischer and N.~Homma, Eds.\hskip 1em plus 0.5em minus
  0.4em\relax Cham: Springer International Publishing, 2017, pp. 665--686.

\bibitem{ATmega2560Medical}
P.~Szakacs-Simon, S.~A. Moraru, and F.~Neukart, ``Signal conditioning
  techniques for health monitoring devices,'' in \emph{2012 35th International
  Conference on Telecommunications and Signal Processing (TSP)}, 2012, pp.
  610--614.

\bibitem{microECC}
\BIBentryALTinterwordspacing
K.~MacKay, ``micro-ecc: Ecdh and ecdsa for 8-bit, 32-bit, and 64-bit
  processors,'' Github Repository. [Online]. Available:
  \url{https://github.com/kmackay/micro-ecc}
\BIBentrySTDinterwordspacing

\bibitem{Precomputation:LowCostSig:Ateniese:NDSS2013}
G.~Ateniese, G.~Bianchi, A.~Capossele, and C.~Petrioli, ``{L}ow-cost {S}tandard
  {S}ignatures in {W}ireless {S}ensor {N}etworks: {A} {C}ase for {R}eviving
  {P}re-computation {T}echniques?'' in \emph{Proceedings of the 20th Annual
  Network {\&} Distributed System Security Symposium, NDSS}, ser. NDSS2013, San
  Diego, CA, 24-27 2013.

\end{thebibliography}

\end{document}